\documentclass[
twocolumn,
showpacs,preprintnumbers,amsmath,amssymb,aps,pra,
superscriptaddress
]{revtex4-1}

\usepackage{graphicx}
\usepackage{bm}
\usepackage{amsfonts}
\usepackage{amscd}
\usepackage{amsmath}    
\usepackage{enumerate}

\usepackage{amssymb}
\usepackage{amsthm}

\usepackage[breaklinks]{hyperref}

\newtheorem{theorem}{Theorem}
\newtheorem{lemma}{Lemma}
\newtheorem{corollary}{Corollary}

\newcommand{\bra}[1]{\mbox{$\left\langle #1 \right|$}}
\newcommand{\ket}[1]{\mbox{$\left| #1 \right\rangle$}}
\newcommand{\braket}[2]{\mbox{$\left\langle #1 | #2 \right\rangle$}}

\def\tr{{\rm Tr}}

\def\IC{{\mathbb C}}

\newcommand{\maxnorm}[1]{\left\lVert #1 \right\rVert_{\text{max}}}

\DeclareMathOperator{\myRe}{Re}

\def\dmathX#1#2{
$$\lineskiplimit=1000pt \advance\lineskip by #1\jot 
\mathsurround=0pt \tabskip=0pt plus 1000pt
\everycr{\noalign{\penalty\interdisplaylinepenalty}}
\halign to \displaywidth{
\hfil$\displaystyle{##}$\tabskip=0pt&%
\hfil $\displaystyle{{}##{}}$\hfil &%
\hfil $\displaystyle{{}##{}}$\hfil &%
$\displaystyle{##}$\hfil \tabskip=0pt plus 1000pt minus 1000pt&%
\refstepcounter{equation}\label{##}\llap{(\theequation)}\tabskip=0pt\cr
\noalign{\ifdim \prevdepth>-1000pt \vskip -#1\jot\fi}
#2\crcr}$$}

\begin{document}

\title
{Physical time-energy cost of a quantum process determines its information fidelity}

\author{Chi-Hang Fred Fung}
\email{chffung@hku.hk}
\affiliation{Department of Physics and Center of Theoretical and Computational Physics, University of Hong Kong, Pokfulam Road, Hong Kong}
\author{H.~F. Chau}
\affiliation{Department of Physics and Center of Theoretical and Computational Physics, University of Hong Kong, Pokfulam Road, Hong Kong}

\begin{abstract}
A quantum system can be described and characterized by at least two different concepts, namely, its physical and
informational properties.
Here, we explicitly connect these two concepts, by equating
the time-energy cost which is the product of the largest energy of a Hamiltonian of quantum dynamics and the evolution time, and the entanglement fidelity which is the informational difference between an input state and the corresponding output state produced by a quantum channel characterized by the Hamiltonian.
Specifically, the worst-case entanglement fidelity between the input and output states is exactly the cosine of the channel's time-energy cost (except when the fidelity is zero).
The exactness of our relation makes a strong statement about the intimate connection between information and physics.
Our exact result may also be regarded as a time-energy uncertainty relation 
for the fastest state that achieves a certain fidelity.

\end{abstract}


\maketitle

\section{Introduction}

All information processing tasks are carried out by physical systems~\cite{Landauer1991}.
Given a quantum system with certain eigen-energies evolved for a fixed amount of time, what can we say about its information processing characteristics?
Inevitably, the information processing power of a physical system is limited by its physical resources including in particular the evolution time and energies~\cite{Lloyd2000}.
This theme on the relation between the physics and information of computing devices has been intensively investigated since the proposal of the Laudauer's principle in 1961~\cite{Landauer1961}.
Intuitively, the more energy and time used, the more informational work such as flipping and erasing a logical state can be done.
In quantum mechanics, a closed system may be characterized by a (time-independent) Hamiltonian $H$ whose eigenvalues are the energies of the system.
When the system evolves for a time period $t$, its initial state $\ket{\psi_\text{i}}$ will transform unitarily, according to the Schr\"{o}dinger equation, to the final state
$\ket{\psi_\text{f}} = U \ket{\psi_\text{i}}$
where
$U=\exp(-i Ht/\hbar)$.
Consider, for example, the rotation of a qubit about the $X$-axis by 
a unitary transformation $U=\exp(-i \omega \sigma_X)$
where $\omega$ is some fixed parameter and $\sigma_X$ is the Pauli $X$ matrix.
If $\omega=\pi/2$, it is a bit flip operation changing $\ket{0}$ to $-i\ket{1}$; other values of $\omega$ may be regarded as a partial bit flip.
We can implement this operation using a Hamiltonian $H=\hbar \omega \sigma_X /t$ with energies $\pm \hbar \omega/t$ evolved for $t$ amount of time.
Thus, in this case, the product of the time and energy, $\pm \hbar \omega$, gives the amount of informational work (complete or partial bit flip depending on $\omega$) done by the system. 
This motivates the consideration of the time-energy product as a measure of the physical resource in this paper and in previous studies~\cite{Chau2011,Fung:2013:Time-energy,Uzdin:2013:time-energy,Fung:2014:TE:Measurements}.
Also, this product form appears in 
time-energy uncertainty relations (TEURs)~\cite{Mandelstam1945,Bhattacharyya1983,Anandan1990,Uhlmann1992,Vaidman1992,Pfeifer1993,Margolus1996,Margolus1998,Chau2010} as explained below.
On the other hand, the informational difference between two quantum states is often measured by the trace distance and the fidelity.

We investigate in this paper the connection between physics and information in 
the setting of quantum dynamics.
Research in the same theme under the setting of thermodynamics has also been investigated~\cite{Landauer1991,Oppenheim2002,delRio:2011:negativeentropy}.
These studies often deal with heat dissipation and entropy changes.
Research in both settings has been actively carried out.

Previous works based on quantum dynamics have resulted in 
TEURs
in the study of quantum speed limit~\cite{Mandelstam1945,Bhattacharyya1983,Anandan1990,Uhlmann1992,Vaidman1992,Pfeifer1993,Margolus1996,Margolus1998,Chau2010}.
Many TEURs 
often take the form of the product of (functions of) the eigen-energies of the quantum system and the evolution time being greater than or equal to 
the fidelity.
These TEURs involving physical and informational properties suggest a connection between them, but the connection is weak since the TEURs are inequalities.
In this paper, we discover an equality relation between the time-energy product and fidelity and thus this establishes a strong link between physics and information.

To be more specific about how the eigen-energies of a Hamiltonian are related to the fidelity, let us consider
a TEUR which is dependent on the energy spread%
~\cite{Bhattacharyya1983,Uhlmann1992,Pfeifer1993}:
given a system characterized by a time-independent Hamiltonian 
$H$,
the time $t$ needed to evolve an initial state $\rho$ to a final state $\rho'$ is
\begin{align}
\label{eqn-time-energy-delta-E}
t \Delta E \ge \hbar \cos^{-1}(F(\rho,\rho'))
\end{align}
where 
$F(\rho,\rho')\equiv\tr \sqrt{\rho^{1/2} \rho' \rho^{1/2}}$
is the fidelity between two mixed quantum states $\rho$ and $\rho'$~\cite{Jozsa1994,*Uhlmann1976},
and $\Delta E=\sqrt{\tr (H^2 \rho)-[\tr(H \rho)]^2}$ is the standard deviation of the system energy.
We will prove in this paper an equality of the same spirit with the left-hand side (LHS) corresponding to a similar notion of time-energy cost and the right-hand side (RHS) the entanglement fidelity.
We elaborate on these two quantifications next.

The notion of time-energy cost we consider essentially corresponds to the product of the largest eigen-energy and the evolution time.
Precisely, 
the time-energy cost of a unitary matrix 
$U \in \IC^{r \times r}$ is defined as~\cite{Chau2011}:
\begin{align}
\label{eqn-definition-maxnorm-for-U}
\maxnorm{U}&=\max_{1 \le j \le r}  |\theta_j|
\end{align}
where $U$ has eigenvalues $\exp(-i E_j t/\hbar)\equiv \exp(-i \theta_j)$ for $j=1,\dots,r$ and $E_j$ are the eigenvalues of the time-independent Hamiltonian $H$~\footnote{We remark that our previous works~\cite{Chau2011,Fung:2013:Time-energy} consider more general measures by taking linear combinations of $|\theta_j|$'s.  Here, we only consider the maximum $|\theta_j|$.}.
We take the convention that all angles are in the range $(-\pi,\pi]$.
Thus, the required energy of the Hamiltonian $H$ to implement $U$ in $t$ amount of time is 
$\maxnorm{U} \hbar/t$.
In essence, time and energy are a trade-off against each other in the sense that the same $U$ may be implemented with a high-energy $H$ evolved for a short time period or a low-energy $H$ for a long time period.
The concept of time-energy cost has been naturally extended to quantum channels by considering the most efficient unitary transformation in a larger Hilbert space embedding a given quantum channel~\cite{Fung:2013:Time-energy}.
We denote the time-energy cost for a quantum channel $\mathcal K$ by $\maxnorm{\mathcal K}$ (which will be defined later in Eq.~\eqref{eqn-energy-measure-general-channel}).

The informational aspect of a quantum process is often captured by fidelity.
Here, we consider a fidelity measure for a quantum channel $\mathcal K$ which we call the minimum entanglement fidelity~\cite{Schumacher:1996:entanglement}:
\begin{align}
\label{eqn-def-fidelity-channel}
F_\text{min}({\mathcal K})
\equiv
\min_{\ket{\Psi}}
F\big(\ket{\Psi}_{AB}\bra{\Psi},(I_A \otimes {\mathcal K}_B)(\ket{\Psi}_{AB}\bra{\Psi})\big).
\end{align}
Here, $\ket{\Psi}$ is a joint state of systems $A$ and $B$, and the channel $\mathcal K$ is only applied to system $B$ in the second term on the RHS.
In essence, this measure $F_\text{min}$ involves comparing the channel input and output and obtaining the input with the minimum fidelity~\footnote{Note that it is not necessary to consider mixed states in the minimization because of the joint concavity of fidelity
$F(\sum_i p_i \rho_i, \sum_i p_i \sigma_i) \ge \sum_i p_i F(\rho_i,\sigma_i)$.}.
Note that we allow the input to be entangled with ancillary system $A$ of any dimension and the comparison is done {\it with that system included}.
The ancillary system gives greater ability to distinguish states.
Fidelity is often used to characterize the informational properties of many quantum information processing 
tasks including
quantum key distribution (as
a security measure~\cite{Konig:2007:AccessibleInformation,Ben-Or:2005:Composable}),
state discrimination (as the inconclusive probability~\cite{Ivanovic:1987:USD,Dieks:1988:USD,Peres:1988:USD}), and
information transmission 
(as a parameter for quantifying quantum Fano-type inequalities and quantum channel capacities~\cite{Schumacher:1996:entanglement,Adami:1997:capacity}).

Information processing is ultimately carried out by physical systems~\cite{Landauer1991}.
It makes sense that the processing power is related to the physical resources used.
This paper provides a partial answer in this direction in terms of $F_\text{min}({\mathcal K})$ and $\maxnorm{\mathcal K}$.

Main result ---
In this paper, we prove that for any quantum channel $\mathcal K$, 
its physical aspect (time-energy cost) is directly related to its informational aspect (fidelity):
\begin{theorem}
\label{thm-main-theorem}
{\rm
\begin{align}
\label{eqn-main-theorem}
F_\text{min}({\mathcal K})
=
\max(
\cos \maxnorm{\mathcal K},0
) .
\end{align}
Here the fidelity $F_\text{min}({\mathcal K})$ is defined in Eq.~\eqref{eqn-def-fidelity-channel} and 
the time-energy cost $\maxnorm{\mathcal K}$ is defined in Eq.~\eqref{eqn-energy-measure-general-channel}.
}
\end{theorem}
This theorem shows that the more time-energy cost incurred by the quantum channel (or quantum process), the less similar are the (worst-case) channel input and output states.
This exact relation may also be considered as a TEUR (see Sec.~\ref{sec-TEUR}).
Note 
the similarity between Eqs.~\eqref{eqn-main-theorem} and \eqref{eqn-time-energy-delta-E}, and
the similarity between $t \Delta E /\hbar$ here and $\theta_j$ in Eq.~\eqref{eqn-definition-maxnorm-for-U}.
However, unlike most TEURs such as Eq.~\eqref{eqn-time-energy-delta-E}, our relation completely separates the physical aspect and the informational aspect in that the time-energy term $\maxnorm{\mathcal K}$ is independent of the channel state.

In addition, our result shares the same spirit as an earlier observation~\cite{Chau2011} that
the time-energy cost tightly bounds the fidelity (in the form of Bures angle).
Given two states $\ket{\psi_1}$ and $\ket{\psi_2}$ separated by Bures angle $\chi=\cos^{-1} \lvert \braket{\psi_1}{\psi_2} \rvert$, 
all unitary $U$ satisfying $\ket{\psi_2}=U \ket{\psi_1}$ must have $\maxnorm{U} \ge \chi$ (see section 3 of \cite{Chau2011}).

\section{Preliminaries}
We consider a quantum channel mapping $n$-dimensional density matrices to $n$-dimensional density matrices:
\begin{align}
{\mathcal K}(\rho)=\sum_{j=1}^d K_j \rho K_j^\dag,
\end{align}
where $K_j \in \IC^{n\times n}$ are the Kraus operators.
In this paper, we only consider finite dimensional systems.
The time-energy cost for quantum channel $\mathcal K$ is defined as~\cite{Fung:2013:Time-energy}
\dmathX2{
\maxnorm{\mathcal{K}} &\equiv& \min_U & \maxnorm{U}  
&eqn-energy-measure-general-channel\cr
&&
\text{s.t.} &
\mathcal{K}(\rho) = \tr_{C} [ U_{CB} (\ket{0}_{C}\bra{0} \otimes  \rho_B) U_{CB}^\dag ]
\: \forall \rho,
\cr
}
where channel $\mathcal{K}$ acts on quantum state $\rho$ in system $B$ and the unitary extension $U_{CB}$ includes ancillary system $C$ prepared in a standard state.
We emphasize that ancillary system $A$ in Eq.~\eqref{eqn-def-fidelity-channel} is a different system from ancillary system $C$ here.
This time-energy cost admits the following general solution~\cite{Fung:2014:TEproof}.
\begin{theorem}
\label{thm-TE-general-solution}
{\rm
\begin{align}
\label{eqn-TE-general-solution}
\maxnorm{\mathcal{K}} =
\cos^{-1}
\left[
\max_{\mathbf v}
\frac{1}{2} \lambda_\text{min} \left( K_{\mathbf v} + K_{\mathbf v}^\dag \right)
\right]
\end{align}
where ${\mathbf v} \in \IC^d$ has unit norm,
$K_{\mathbf v}=\sum_{j=1}^d v_j K_j$, $\lambda_\text{min}(\cdot)$ denotes the minimum eigenvalue of its argument,
and
we take the convention that $\cos^{-1}$ returns an angle in the range $[0,\pi]$.
}
\end{theorem}
For a class of channels which includes the depolarizing channel, a simple closed-form solution has been found~\cite{Fung:2013:Time-energy}.

\section{Proof of our main result}

Here, we prove our main result, Theorem~\ref{thm-main-theorem}.
We first compute the fidelity in Eq.~\eqref{eqn-def-fidelity-channel} for a fixed input state $\ket{\Psi}_{AB}$:
\begin{align}
&
F\big(\ket{\Psi}_{AB}\bra{\Psi},(I_A \otimes {\mathcal K}_B)(\ket{\Psi}_{AB}\bra{\Psi})\big)
\nonumber
\\
=&
\sqrt{
\sum_i \lvert \bra{\Psi} (I_A \otimes K_i) \ket{\Psi} \rvert^2
}
\nonumber
\\
=&
\sqrt{
\sum_i \lvert \tr (\rho_B K_i) \rvert^2
}
\equiv F_e(\rho_B,{\mathcal K})
\label{eqn-fidelity-form1}
\end{align}
where $\rho_B=\tr_A(\ket{\Psi}_{AB}\bra{\Psi})$.
The fidelity $F_e$ is known as the {\em entanglement fidelity} of the channel $\mathcal K$~\cite{Schumacher:1996:entanglement}, and
is independent of which purification $\ket{\Psi}_{AB}$ is used.
Another way to express $F_e$ is
\begin{align}
F_e(\rho_B,{\mathcal K})=\max_{\bf w} \left\lvert \sum_i w_i \tr (\rho_B K_i) \right\rvert
\label{eqn-fidelity-form2}
\end{align}
where ${\bf w} \in \IC^d$ has unit norm.
This follows either from the Cauchy-Schwarz inequality with the solution
\begin{equation}
\label{eqn-fidelity-solution-for-w}
w_i=
\frac{
\tr^\dag (\rho_B K_i)
}{
\sqrt{
\sum_i \lvert \tr (\rho_B K_i) \rvert^2
}
}
\end{equation}
or from the purification definition of fidelity using Uhlmann's theorem~\cite{Uhlmann1976,*Jozsa1994}.

Using Eq.~\eqref{eqn-fidelity-form2},
the minimum entanglement fidelity of the channel is
\begin{equation}
F_\text{min}({\mathcal K}) = 
\min_{\rho_B}
\max_{\bf w} \left\lvert \sum_{i=1}^d w_i \tr (\rho_B K_i) \right\rvert
\label{eqn-Fmin-form2}
\end{equation}
and we denote the optimal solution by $\tilde{\rho}_B$ and $\tilde{\bf w}$ which is given by Eq.~\eqref{eqn-fidelity-solution-for-w} with $\rho_B \rightarrow \tilde{\rho}_B$.
Furthermore, using Eq.~\eqref{eqn-fidelity-form1}, we have
\begin{align}
F_\text{min}({\mathcal K}) = 
\sqrt{
\sum_{i=1}^d \lvert \tr (\tilde{\rho}_B K_i) \rvert^2
} .
\label{eqn-Fmin-form1}
\end{align}

\begin{lemma}
\label{lemma-not-state-exists}
{\rm
There does not exist a state $\ket{\psi'}_B$ such that 
$$
0 \leq \myRe \left( \bra{\psi'} \sum_i \tilde{w}_i K_i \ket{\psi'} \right) < F_\text{min}({\mathcal K}) .
$$
}
\end{lemma}
\begin{proof}
We prove by contradiction.
Suppose that such a state $\ket{\psi'}_B$ exists.
We form a new state $\rho_B = (1-\alpha)\tilde{\rho}_B + \alpha \ket{\psi'}_B \bra{\psi'}$ where $\alpha >0$ is a small parameter and calculate the squared fidelity for this state using Eq.~\eqref{eqn-fidelity-form1}:
\begin{align*}
F_e^2(\rho_B,{\mathcal K})
=&
\sum_i 
\Big\{
(1-\alpha)^2 \lvert \tr (\tilde{\rho}_B K_i) \rvert^2
+
\alpha^2 
\lvert \bra{\psi'} K_i \ket{\psi'} \rvert^2
\\
&
\phantom{x}
+2(1-\alpha)\alpha \myRe \left( \tr^\dag (\tilde{\rho}_B K_i) \bra{\psi'} K_i \ket{\psi'} \right)
\Big\} .
\end{align*}
For $\alpha\rightarrow 0$, all the second-order terms become negligible and 
the change as a function of $\alpha$ is
$$
\frac{\partial F_e^2}{2 \partial\alpha}=
\sum_i 
\Big\{
\myRe \left( \tr^\dag (\tilde{\rho}_B K_i) \bra{\psi'} K_i \ket{\psi'} \right)
-
\lvert \tr (\tilde{\rho}_B K_i) \rvert^2
\Big\}
.
$$
Note that the second term on the right is $F_\text{min}^2({\mathcal K})$ (see Eq.~\eqref{eqn-Fmin-form1}).
With the help of Eq.~\eqref{eqn-fidelity-solution-for-w},
the first term on the right can be expressed as
$$
\sqrt{
\sum_j \lvert \tr (\tilde{\rho}_B K_j) \rvert^2
}
\myRe 
\left(
\sum_i 
\left( \tilde{w}_i \bra{\psi'} K_i \ket{\psi'} \right)
\right).
$$
This means that if the claimed state $\ket{\psi'}$ exists, $\partial F_e^2/\partial\alpha < 0$ and this contradicts with that fact that when $\alpha=0$, $F_e^2=F_\text{min}^2({\mathcal K})$ which is already the minimum and cannot become smaller with any other states.
\end{proof}
Note that Lemma~\ref{lemma-not-state-exists} does not cover the trivial case of $F_\text{min}({\mathcal K})=0$.
\begin{lemma}
\label{lemma-exists-pure-state-for-Fmin}
{\rm
There exists a state $\ket{\psi'}_B$ such that 
\begin{equation*}
\bra{\psi'} \sum_i \tilde{w}_i K_i \ket{\psi'} = F_\text{min}({\mathcal K}) .
\end{equation*}
}
\end{lemma}
\begin{proof}
We first recognize that the LHS is a numerical range.
Recall that the numerical range of an operator $K \in \IC^{n \times n}$ is defined as
$$
W(K)=\{ \bra{\psi'} K \ket{\psi'} : \ket{\psi'} \in \IC^n,  \braket{\psi'}{\psi'}=1 \} .
$$
Any numerical range is convex in the sense that if
$\bra{\psi'} K \ket{\psi'}$ and
$\bra{\psi''} K \ket{\psi''}$ are in $W(K)$,
then for $0 \le \alpha \le 1$,
$\alpha \bra{\psi'} K \ket{\psi'}+(1-\alpha)\bra{\psi''} K \ket{\psi''}$ 
is in $W(K)$.

Note that $F_\text{min}({\mathcal K})$ in Eq.~\eqref{eqn-Fmin-form2} can be expressed as
$$
F_\text{min}({\mathcal K})=\tr \left( \tilde{\rho}_B \sum_i \tilde{w}_i K_i \right) ,
$$
which is real.
Since the numerical range of $\sum_i \tilde{w}_i K_i$ is convex and any mixed state is a linear combination of pure states,
for any $\rho_B$, there exists $\ket{\psi}$ such that
$\tr \Big( \rho_B \sum_i \tilde{w}_i K_i \Big)=
\bra{\psi} \sum_i \tilde{w}_i K_i \ket{\psi}
$.
This completes the proof.
\end{proof}
\begin{corollary}
\label{cor-no-state-on-left2}
{\rm
If $F_\text{min}({\mathcal K})>0$,
there does not exist a state $\ket{\psi'}_B$ such that 
\begin{equation*}
\myRe \left( \bra{\psi'} \sum_i \tilde{w}_i K_i \ket{\psi'} \right) < F_\text{min}({\mathcal K}) .
\end{equation*}
}
\end{corollary}
\begin{proof}
This follows from Lemma~\ref{lemma-not-state-exists}, 
Lemma~\ref{lemma-exists-pure-state-for-Fmin}, and
the fact that any numerical range is convex.
\end{proof}
\begin{lemma}
\label{lemma-Fmin-min-over-states}
{\rm
If $F_\text{min}({\mathcal K})>0$,
\begin{equation}
F_\text{min}({\mathcal K}) = \min_{\ket{\psi}} 
\myRe \left( \bra{\psi} \sum_i \tilde{w}_i K_i \ket{\psi} \right)
\label{eqn-Fmin-form-min-only}
\end{equation}
where the minimization is over all normalized pure states $\ket{\psi}$.
}
\end{lemma}
\begin{proof}
This follows from 
Lemma~\ref{lemma-exists-pure-state-for-Fmin} and Corollary~\ref{cor-no-state-on-left2}.
\end{proof}
\begin{theorem}
\label{thm-Fmin-ge-maxnorm}
{\rm
$$
F_\text{min}({\mathcal K}) \ge
\cos \maxnorm{\mathcal K}.
$$
}
\end{theorem}
\begin{proof}
First note that Eq.~\eqref{eqn-TE-general-solution} can be written as
\begin{align}
\cos \maxnorm{\mathcal K}
&=
\max_{\mathbf v}
\min_{\ket{\psi}} 
\frac{1}{2} \bra{\psi} \left( K_{\mathbf v} + K_{\mathbf v}^\dag \right) \ket{\psi}
\nonumber
\\
&=
\max_{\mathbf v}
\min_{\ket{\psi}} 
\myRe \left( \bra{\psi} \sum_i v_i K_i \ket{\psi} \right).
\label{eqn-maxnorm-form-Re1}
\end{align}

On the other hand,
from Eq.~\eqref{eqn-Fmin-form2}, we have
\begin{align*}
F_\text{min}({\mathcal K}) 
&= 
\min_{\rho_B}
\max_{\bf w} \left\lvert \sum_i w_i \tr (\rho_B K_i) \right\rvert
\\
&\ge
\max_{\bf w}
\min_{\rho_B}
\left\lvert \sum_i w_i \tr (\rho_B K_i) \right\rvert
\\
&=
\max_{\bf w}
\min_{\ket{\psi}} 
\left\lvert \bra{\psi} \sum_i w_i K_i \ket{\psi} \right\rvert
\\
&\ge
\max_{\mathbf w}
\min_{\ket{\psi}} 
\myRe \left( \bra{\psi} \sum_i w_i K_i \ket{\psi} \right)
\end{align*}
where the first inequality in the second line is due to a general inequality known as the max-min inequality (see, e.g., Ref.~\cite{Boyd:2004}), the third line is due to the convexity of numerical ranges (see the proof of Lemma~\ref{lemma-exists-pure-state-for-Fmin}), and the fourth line is because $|x| \ge \myRe (x)$ for all $x$.
Comparing with Eq.~\eqref{eqn-maxnorm-form-Re1} proves the claim.
\end{proof}
\begin{theorem}
\label{thm-Fmin-equal-maxnorm}
{\rm
If $F_\text{min}({\mathcal K})>0$,
$$
F_\text{min}({\mathcal K})=
\cos \maxnorm{\mathcal K}.
$$
}
\end{theorem}
\begin{proof}
Note that Eq.~\eqref{eqn-Fmin-form-min-only} is less than or equal to Eq.~\eqref{eqn-maxnorm-form-Re1}, giving
$$
F_\text{min}({\mathcal K})
\le
\cos \maxnorm{\mathcal K}.
$$
This together with Theorem~\ref{thm-Fmin-ge-maxnorm} gives the result.
\end{proof}

\begin{proof}
[{\bf Proof of Theorem~\ref{thm-main-theorem}}]
Theorem~\ref{thm-main-theorem} follows from Theorems~\ref{thm-Fmin-ge-maxnorm} and \ref{thm-Fmin-equal-maxnorm}.
\end{proof}

\section{Example}

The quantum depolarizing channel acting on 
$n \times n$
density matrices is defined as
\begin{align*}
{\mathcal K}_\text{D} (\rho) & \equiv q \rho + (1-q) \frac{I}{n}
\end{align*}
where 
complete positivity requires that
$-1/(n^2-1) \leq q \leq 1$~\cite{King2003}.
The minimum entanglement fidelity can be achieved
with the
input state $\ket{\Psi}_{AB}=\sum_{i=0}^{n-1} \ket{ii}_{AB} /\sqrt{n}$.
The output state is
$$
(I \otimes {{\mathcal K}_\text{D}})(\ket{\Psi}\bra{\Psi})
=
\sum_{i,j=0}^{n-1} \frac{q}{n} \ket{ii}_{AB}\bra{jj} + \frac{1-q}{n^2} \ket{i}_A\bra{j} \otimes I_B
$$
and the minimum entanglement fidelity 
calculated using Eq.~\eqref{eqn-def-fidelity-channel} 
is
$$
F_\text{min}({\mathcal K}_\text{D})=\sqrt{q+\frac{1-q}{n^2}}.
$$
The time-energy cost has been proved in a previous work (see Eq.~(59) of Ref.~\cite{Fung:2013:Time-energy}) to be 
$$
\maxnorm{{\mathcal K}_\text{D}}=\cos^{-1}\sqrt{q+\frac{1-q}{n^2}},
$$
which can be checked to be consistent with Theorem~\ref{thm-main-theorem} for the entire range of $q$.
On the other hand, given the non-zero minimum entanglement fidelity of a channel, we can easily infer its time-energy cost using Theorem~\ref{thm-main-theorem}.

We can also compute the fidelity without entanglement using channel input $\ket{0}\bra{0}$ which produces the output $q\ket{0}\bra{0}+(1-q)I/n$.
The fidelity is thus $\sqrt{q+(1-q)/n}$.
Hence, the fidelity without entanglement does not correspond to the time-energy cost of the channel in general.

\section{Connection with time-energy uncertainty relation}

\label{sec-TEUR}

We begin by considering closed systems.
The channel $\mathcal K$ for a closed system is a unitary transformation $U$ of dimensions $n \times n$ with eigenvalues
$\exp(i \theta_j), j=1,\dots,n$.
Suppose that $-\pi/2 \le \theta_j\le \pi/2$ for all $j$.
Then, in this case, 
Eq.~\eqref{eqn-TE-general-solution} of
Theorem~\ref{thm-TE-general-solution} 
simplifies to
\begin{align*}
\maxnorm{\mathcal{K}} &=
\cos^{-1}
\left[
\max_{\gamma}
\frac{1}{2} \lambda_\text{min} \left( e^{i \gamma} U + e^{-i \gamma}U^\dag \right)
\right]
\\
&=
\frac{\theta_\text{max}-\theta_\text{min}}{2}
\hspace{2.3cm}\text{(for $|\theta_j| \le \pi/2$)}
\end{align*}
where 
$\theta_\text{max}=\max_j \theta_j=\maxnorm{U}$ and $\theta_\text{min}=\min_j \theta_j$.
This can be easily shown by noting that $\gamma$ is chosen so that the two left-most eigenvalues of $e^{i \gamma} U$ (one above and one below the real line) are a complex conjugate of each other.

Since $U=\exp(-i H t/\hbar)$ for some Hamiltonian $H$ with eigen-energies $E_j, j=1,\dots,n$,
we have 
\begin{equation}
\label{eqn-closed-system-TE-product-1}
\maxnorm{\mathcal{K}} = \frac{(E_\text{max}-E_\text{min}) t}{2 \hbar}
\end{equation}
where $E_\text{max}$ and $E_\text{min}$ are the maximum and minimum eigen-energies. And according to Theorem~\ref{thm-main-theorem}, there is an input state having an entanglement fidelity $F$ with the corresponding channel output state given by 
$
F=\cos \maxnorm{\mathcal K}
$.
Furthermore, all other input states have an entanglement fidelity no smaller than this.
We may consider evolving the system with this Hamiltonian.
As the system evolves, $t$ increases from zero and $F$ decreases from one.
Thus, the fastest state for this Hamiltonian that achieves an entanglement fidelity $F$ takes time
\begin{equation}
\label{eqn-closed-system-my-TEUR-1}
t=
\frac{2 \hbar \cos^{-1}(F)}{E_\text{max}-E_\text{min}}.
\end{equation}
In particular, the minimum orthogonalization time is
\begin{equation}
\label{eqn-closed-system-my-TEUR-2}
t_\text{orth}=
\frac{\pi \hbar}{E_\text{max}-E_\text{min}}.
\end{equation}
This means that no state can be orthogonalized faster than this time $t_\text{orth}$.
Equations~\eqref{eqn-closed-system-my-TEUR-1} and \eqref{eqn-closed-system-my-TEUR-2} may be regarded as TEURs for the fastest states for a Hamiltonian implementing a unitary channel.

The TEUR for the fastest state in Eq.~\eqref{eqn-closed-system-my-TEUR-2} may be used as a reference for the orthogonalization time of a given input state computed using a standard TEUR.
For example, 
Chau~\cite{Chau2010} proposed a TEUR that gives the orthogonalization time for the state 
$\frac{1}{\sqrt{2}}( \ket{-\mathcal E}+\ket{\mathcal E})$
where $\ket{\pm \mathcal E}$ are the eigen-states of the Hamiltonian with corresponding eigen-energies $\pm \mathcal E$.  This time is computed to be $\hbar/(A \mathcal E)$ where $A \approx 0.724611$.
On the other hand, Eq.~\eqref{eqn-closed-system-my-TEUR-2} gives $\pi \hbar/(2 \mathcal E)$ which is larger.
This means that 
the TEUR in Ref.~\cite{Chau2010} is not tight for that particular state.

We may extend this concept to general quantum channels.
Given a channel $\mathcal K$, by definition its time-energy value $\maxnorm{\mathcal K}$ is the smallest of the time-energy values of all unitary extensions.
Thus, Eq.~\eqref{eqn-closed-system-TE-product-1} also holds for general quantum channel $\mathcal K$ with $E_\text{max}$ and $E_\text{min}$ being the maximum and minimum energies of the Hamiltonian corresponding to the best unitary extension, and $t$ being the evolution time of the Hamiltonian to result in the channel $\mathcal K$,
provided that 
the eigen-angles $\theta_j$'s of the unitary extension satisfy
$|\theta_j|\le \pi/2$.
Similarly, Eqs.~\eqref{eqn-closed-system-my-TEUR-1} and \eqref{eqn-closed-system-my-TEUR-2} applies to this unitary extension.
For other suboptimal unitary extensions satisfying $|\theta_j|\le \pi/2$,
we have
\begin{equation}
\maxnorm{\mathcal{K}} \le \frac{(E_\text{max}-E_\text{min}) t}{2 \hbar}
\end{equation}
where $E_\text{max}$ and $E_\text{min}$ are the energies of the corresponding Hamiltonian.

\section{Conclusions}

We established an exact relation between the physical aspect of any quantum process and its informational aspect.
The time-energy cost of a quantum channel has an interpretation of being the amount of physical resources incurred for performing the action of the channel.
Intuitively the larger this amount, the more action is done on its input state, and our result in Theorem~\ref{thm-main-theorem} confirms this intuition strongly since our relation in the theorem is exact.
Our relation may also be regarded as a 
TEUR
for the fastest state that achieves a certain fidelity.
We believe that our exact relation sheds new light on the understanding of the limit on information processing from a 
quantum dynamical perspective.

\section*{Acknowledgments}%
This work is supported in part by
RGC under Grant No. 700712P of the HKSAR Government.

\bibliographystyle{apsrev4-1}

\bibliography{paperdb}

\end{document}